\documentclass[11pt]{article}
\usepackage{amsthm}
\usepackage{amssymb}
\usepackage{amsmath}
\usepackage{todonotes}
\usepackage[margin=1in]{geometry}
\usepackage[utf8]{inputenc}
\usepackage{hyperref}
\usepackage{pbox}
\usepackage{framed}
\usepackage{mathtools}
\usepackage{empheq} %% https://tex.stackexchange.com/questions/109900/how-can-i-box-multiple-aligned-equations
\usepackage{thmtools}
\usepackage{thm-restate}

\theoremstyle{plain}
\newtheorem{thm}{Theorem}
\numberwithin{thm}{section}

\newtheorem{prop}[thm]{Proposition}
\newtheorem{claim}[thm]{Claim}
\newtheorem{lem}[thm]{Lemma}

\theoremstyle{definition}
\newtheorem{df}{Definition}

\numberwithin{df}{section}

\theoremstyle{remark}
\newtheorem*{rem}{Remark}

\newcommand{\CSP}{\operatorname{CSP}}

\newcommand{\poly}{\operatorname{poly}}

\newcommand{\Ham}{\operatorname{Ham}}

\newcommand{\ar}{\operatorname{ar}}

\newcommand{\meas}{\operatorname{meas}}

\title{\textbf{Bridging between 0/1 and Linear Programming via Random Walks}}
\author{Joshua Brakensiek\thanks{Department of Computer Science, Stanford University, Stanford, CA. Email: {\tt jbrakens@stanford.edu}. Some of this work was done when the author was at Carnegie Mellon University. Research supported in part by NSF CCF-1526092, and an NSF Graduate Research Fellowship.} \and Venkatesan Guruswami\thanks{Computer Science Department, Carnegie Mellon University, Pittsburgh, PA 15213. Email: {\tt venkatg@cs.cmu.edu}. Research supported in part by NSF grants CCF-1422045 and CCF-1526092.}}
\date{}
\begin{document}

\maketitle
\thispagestyle{empty}

\begin{abstract}
Under the Strong Exponential Time Hypothesis, an integer linear program with $n$ Boolean-valued variables and $m$ equations cannot be solved in $c^n$ time for any constant $c < 2$. If the domain of the variables is relaxed to $[0,1]$, the associated linear program can of course be solved in polynomial time. In this work, we give a natural algorithmic bridging between these extremes of $0$-$1$ and linear programming. Specifically, for any subset (finite union of intervals) $E \subset [0,1]$ containing $\{0,1\}$, we give a random-walk based algorithm with runtime $O_E((2-\text{measure}(E))^n\poly(n,m))$ that finds a solution in $E^n$ to any  $n$-variable linear
 program with $m$ constraints that is feasible over $\{0,1\}^n$. Note that as $E$ expands from $\{0,1\}$ to $[0,1]$, the runtime improves smoothly from $2^n$ to polynomial. 
 
 \smallskip
Taking $E = [0,1/k) \cup (1-1/k,1]$ in our result yields as a corollary a randomized $(2-2/k)^{n}\poly(n)$ time algorithm for
$k$-SAT. While our approach has some high level resemblance to Sch\"{o}ning's beautiful algorithm, our general algorithm is based on a more sophisticated random walk that incorporates several new ingredients, such as a multiplicative potential to measure progress, a judicious choice of starting distribution, and a time varying distribution for the evolution of the random walk that is itself computed via an LP at each step (a solution to which is guaranteed based on the minimax theorem).  Plugging the LP algorithm into our earlier polymorphic framework yields fast exponential algorithms for any CSP (like $k$-SAT, $1$-in-$3$-SAT, NAE $k$-SAT) that admit so-called ``threshold partial polymorphisms."

\end{abstract}

\newpage

\def\exponent{c}
\section{Introduction}\label{sec:intro}

The study of exponential time algorithms for NP-hard problems has been a thriving area of research, and clever algorithms much faster than naive
brute-force methods have been devised for many problems. For instance,
the canonical NP-complete problem $3$-SAT admits an elegant
$O^\ast((4/3)^n)$ time randomized algorithm~\cite{DBLP:conf/focs/Schoning99} (we use
$O^\ast(\cdot)$ to hide $\text{poly}(n)$ factors, and $n$ is the
number of variables). The runtime of Sch\"{o}ning's algorithm is
$O^\ast((2-2/k)^n)$ for $k$-SAT, and the algorithm has been
derandomized to achieve similar runtime bounds
deterministically~\cite{MoserScheder2011}. Even faster algorithms running
in time $(1.308)^n$ and $(1.469)^n$ are known for $3$-SAT and $4$-SAT
respectively~\cite{Hertli2014}. A survey (albeit not very recent) of exact
exponential time algorithms for NP-complete problems appears as
\cite{Woeginger2003}.

However, there are problems such as general CNF-SAT for which no
algorithm with runtime $c^n$ is known for any constant $c < 2$.  The
strong exponential time hypothesis (SETH)~\cite{IP-seth} asserts that
in fact no such algorithm exists, and indeed that for every $\epsilon >
0$, there exists some $k$ such that $k$-SAT can't be solved in
$(2-\epsilon)^n$ time. Such ``SETH-hardness'' is also known for problems
such as set cover and Not-all-Equal-SAT~\cite{CyganEtAl2016}. The class
of $0$-$1$ integer linear programs, which includes satisfiability, set
cover, and most natural problems with Boolean-valued variables, also
doesn't admit a $c^n$ time algorithm for any $c < 2$. On the other
hand, linear programs where variables are allowed to take values in $[0,1]$
can of course be solved in polynomial time.

Our main result stated below bridges between these extremes, giving an
algorithm with progressively better runtimes as one allows more and
more relaxed values for the Boolean variables. 
As the allowed set $E$ of values shrinks from $[0,1]$ to $\{0,1\}$, the runtime of our algorithm degrades gracefully from polynomial to $2^n$. 
\begin{thm}[Main]
\label{thm:intro-main}
Let $E \subset [0,1]$ be a finite union of intervals that contains
$\{0,1\}$ which is given explicitly: $E = [0, d_1] \cup \cdots \cup [c_k, 1]$. Let $\meas(E)$ be its measure (total length). There is a randomized algorithm that on input an $n$-variable linear
program with $m$ constraints that is feasible over $\{0,1\}^n$, 
runs in time $c_E(2-\meas(E))^{n}\poly(n, m)$, where $c_E$ is a constant depending on the structure of $E$, and with high probability finds a feasible solution to the LP that belongs to $E^n$.
\end{thm}
Fix a constant $k$ and consider $E = E_k = [0,1/k) \cup (1-1/k,1]$. For this $E$, we get a $(2-2/k)^{n} \poly(n,m)$ time algorithm. Solving the natural LP associated\footnote{For each variable $x_i$ in the $k$-SAT instance, specify that $0 \le x_i \le 1$. For each clause $x_{i_1} \wedge x_{i_2} \wedge \cdots x_{i_k}$ specify that $x_{i_1} + x_{i_2} + \cdots + x_{i_k} \ge 1$. For any variable that's negated in a clause, replace $x_{i_j}$ with $1 - x_{i_j}$.} with $k$-SAT in this
way and rounding variables in $[0,1/k)$ to $0$ and those in
$(1-1/k,1]$ to $1$, we get a $k$-SAT algorithm that matches the
runtime of Sch\"{o}ning's celebrated algorithm~\cite{DBLP:conf/focs/Schoning99}, although better algorithms for $k$-SAT are known (e.g., \cite{PPSZ}). The formulation of the problem of
finding a relaxed solution in $E^n$ is, in our opinion, a valuable
contribution in itself, and it further enables a unified treatment of algorithms for $k$-SAT as well as a more general class of constraint satisfaction problems (CSPs) (to be
described shortly).

The algorithm claimed in Theorem~\ref{thm:intro-main} is based on a
random walk approach inspired by Sch\"{o}ning's algorithm.

Even for the case of $E_k$, many new ideas are needed to bring the algorithm to fruition. In particular we use a novel global criterion which guides for making local changes to an evolving assignment. See Section~\ref{subsec:proofoverview} for an overview of our approach.

We note that the work of Impagliazzo, Lovett, Paturi and Schneider~\cite{impagliazzo20140} shows that integer programming can be solved in $2^{(1-\Omega(1/c))n}$ time if the number of constraints is $cn$. The guarantee is incomparable to ours: on the one hand, we allow for many more (even sub-exponentially many) constraints, but on the other hand, we do not find a $0$-$1$ valued solution.

\subsection{Overview of algorithm and analysis}\label{subsec:proofoverview}

Sch\"{o}ning's algorithm for $k$-SAT~\cite{DBLP:conf/focs/Schoning99} consists of the following simple procedure. First, one starts with uniformly random assignment $x \sim \{0, 1\}^n$ to the variables. Then, for $O(n)$ steps one checks if the instance is satisfied (in which case terminate) or else one picks any violated clause and picks a uniformly random variable to flip in that clause.  Sch\"{o}ning's analysis  used the fact that during each random step, the current assignment's Hamming distance to a  ``reference'' solution decreases by $1$ with probability at least $\frac{1}{k}$ and otherwise increases by at most $1$. A simple combinatorial argument can then show that the probability of success is at least $\Omega(\left(\frac{k}{2k-2}\right)^n)$, yielding the $O\left(\left(2 - \frac{2}{k}\right)^n\right)$-time randomized algorithm.

Our algorithm for bridging integer and linear programming has Sch\"{o}ning's algorithm as a blueprint, but many additional insights were needed at both the conceptual and technical levels. A primary challenge is that, one needs to resolve that the variables for our integer/linear program can take on a continuum of values compared to only $2$ (or a finite number) in Sch\"{o}ning's algorithm. To get around this, we use the fact that our target domain $E$ is expressed as the union of finitely many intervals, and thus we consider the discrete problem of determining which interval each variable belongs to. Since the Cartesian product of intervals is a convex region, when one has a set of intervals selected, one can check if there exists some assignment in those intervals by solving a simple modified linear program.

Even with this reduction to a discrete problem, the ``naive'' extension of Sch\"{o}ning's algorithm and analysis still runs into many issues. For instance, the part of Sch\"{o}ning's random walk where he finds a ``bad'' clause and then flips a uniformly random variable needs to be modified greatly. In particular, it is not clear what a ``bad clause'' with respect to our linear program $Ax \le b$, as a particular choice of intervals could be invalid, even though no individual inequality of $Ax \le b$ rules out the entire region. Note that this shows that local methods cannot suffice for this problem (unlike Sch\"{o}ning's algorithm). Thus, we need to use the entire problem $Ax \le b$ as our failed constraint. Furthermore, when modifying the choice of intervals for the next step of the random walk as we no longer have sparsity, it is not clear that uniformly changing some variable will help at all.

To get around these issues, we employ a potential function $U_x(\sigma)$, where $x \in \{0, 1\}^n$ represents an integral solution to $Ax \le b$ while $\sigma$ represents the current choice of intervals in our random walk. The potential function is multiplicative $U_x(\sigma)$ = $\prod_{i=1}^n U_{x_i}(\sigma_i)$, where $U_{x_i}(\sigma_i)$ takes on the value $1$ if $x_i$ is in the interval $\sigma_i$ and decreases depending on how far $x_i$ is from the interval $\sigma_i$. This choice of potential function was inspired by the analysis of the biased random walks corresponding to the ``Gambler's Ruin'' problem (e.g., \cite{feller1968introduction, ethier2002bounds, lawler2010random}). The choice of the potential needs to be precisely tailored to the set $E$, see (\ref{eq:0}, \ref{eq:1}), but intuitively it captures the probability of success of the random walk from the position $\sigma$. Note that since the whole point is to determine the feasibility of the integer program, we cannot explicitly compute $U_x(\sigma)$. Instead, we simultaneously  optimize $U_x(\sigma)$ over all fractional solutions $x \in [0, 1]^n$ to $Ax \le b$.

This optimization involves thinking of the random walk as a two-player game involving Alice and Bob. Imagine that Alice has some adversarial distribution of integral solutions $x \in \{0, 1\}^n$ to $Ax \le b$. Bob needs to compute a distribution of states $\sigma'$ which are changes to $\sigma$ in one position such that $\mathbb E[U_{x}(\sigma')] \ge U_{x}(\sigma)$. By minimax, it suffices to show for any such distribution of strategies, Bob has a deterministic response which causes the potential to increase or stay the same, which results in a simple calculation. This proves that Bob has a random strategy to keep the expected value of the potential as a monovariant no matter what solution Alice is thinking about. We show that this strategy can be explicitly computed by solving a convex optimization problem with a separation oracle.

By itself, knowing that the potential $U_{x}(\sigma)$ is non-decreasing is not too helpful. The key observation is that $U_{x}(\sigma)$ is not staying constant during the random walk, sort of like how in Sch\"{o}ning's algorithm there is at least a $1/k$ chance of making progress. In fact $U_{x}(\sigma')/U_x(\sigma) \not\in (1/\tau, \tau)$ for some $\tau > 1$ we call the \emph{traction}. This ``traction'' is all we need to guarantee that the random walk, originating from state $\sigma$, will succeed with probability at least $\Omega(U_{x}(\sigma))$ in polynomial time.

Finally, one needs to be careful with picking an appropriate starting distribution. A judicious choice is made to ensure a stating potential of at least $\approx (2-\meas(E))^{-n}$ which then leads to the $(2-\meas(E)^n) \poly(n)$ running time. See Section~\ref{subsec:proof} for more details.

Using this linear programming paradigm, we can then obtain a variety of exponential time algorithms for many CSPs by (1) rewriting the CSP as an integer program, (2) picking a set $E \subset [0, 1]$ corresponding to the CSP based on its partial polymorphisms, (3) solving the linear program corresponding to the set $E$ using our algorithm, (4) ``round'' the linear program solution back to a solution to the CSP using partial polymorphism. See Section~\ref{sec:PCSP} for more details.

In summary, by overcoming a few conceptual and technical hurdles, we were able to  generalize Sch\"{o}ning's algorithm to solve a much wider range of feasibility problems.

\subsection{Connection to constraint satisfaction}

We now describe our original motivation, relating to fast exponential algorithms for constraint satisfaction problems (CSPs), that led us toward Theorem~\ref{thm:intro-main}. (In hindsight, we view our main result to be of intrinsic interest from an optimization perspective, well beyond the intended application to CSPs.) In the algebraic theory of constraint satisfaction, the tractability of a CSP has been shown to be intimately tied to the \emph{polymorphisms} associated with its defining relations~\cite{Bulatov2005}. A polymorphism for a relation $R$ is a function $f$ that preserves membership in $R$. That is, if $R \subset D^k$, a function $f: D^m \to D$ is a polymorphism if applying $f$ component-wise to an arbitrary sequence of $m$ $k$-tuples from $R$ always leads to a $k$-tuple in $R$. The resulting rich theory can not only explain and predict polynomial time decidability of CSPs, but with appropriate variations of polymorphisms, also explain the complexity of counting, optimization, and promise versions (c.f., the surveys \cite{Chen2009, DBLP:conf/dagstuhl/BartoKW17}).

In fact, rather remarkably the framework can also shed light on the complexity of CSPs that are NP-complete. \emph{Partial} polymorphisms, which are partial functions that preserve membership in the CSP relations whenever they are defined, can explain the runtime of fast exponential algorithms for CSPs, i.e., the smallest\footnote{Assuming the ETH, one cannot have a $(1+\epsilon)^n$ time algorithm for arbitrarily small $\epsilon > 0$.}  $c > 1$ for which one can get a $c^{n+o(n)}$ time algorithm~\cite{DBLP:conf/soda/LagerkvistJNZ13}. In particular, they show that 1-in-3-SAT, a variant of 3-SAT in which a clause is satisfied only if \emph{exactly} one of the variables is true, is essentially the easiest NP-hard problem in terms of the efficiency of an exponential algorithm.  

In a recent paper~\cite{BrakensiekGuruswami2019}, we presented a general framework to deduce polynomial time algorithms for CSPs (in fact, the more general promise CSPs) that admit as polymorphisms a family of \emph{threshold} functions. An example of such a function $f : \{0, 1\}^m \to \{0, 1\}$ satisfies $f(x) = 1$ if and only if the Hamming weight of $x$ is at least some parameter $\ell$. More generally, such functions can take on more than two values and can have multiple phase transitions at different hamming weights. The Promise CSP algorithm was based on rounding the solution to an associated linear program using the polymorphism.  This framework can be applied with our fast exponential algorithm for finding an LP solution in $E^n$ in order to obtain fast exponential algorithms for CSPs that have a sequence of threshold functions as \emph{partial} polymorphisms. Here  $E$ represents the
  interval of input (fractional) Hamming weights on which the polymorphism is
  defined. (In the case when polymorphisms are defined everywhere,
  $E=[0,1]$ so one can just solve the LP efficiently.) The details of this connection appear in Section~\ref{sec:PCSP}.

    The CSPs for which our method leads to fast exponential algorithms, such as $3$-SAT or $1$-in-$3$-SAT (which has an $O(1.0984^n)$ algorithm~\cite{wahlstrom2007algorithms}), often admit even faster exponential time algorithms optimized for the specific CSP. Our LP solver identifies a fairly general sufficient condition that implies fast exponential algorithms for certain CSPs, giving a principled reason for their existence even if the exact runtime is not optimized.

  \iffalse
  More  generally, the focus in \cite{BrakensiekGuruswami2019} was on developing a framework to
  exploit structured symmetric polymorphisms to efficiently map solutions to linear
  programs, linear equations, or even a combination thereof, to CSP
  solutions. A key innovation was the way different techniques
  (solving LPs and linear systems) can be combined to obtain CSP algorithms
  based on a large class of polymorphisms. In
  contrast, our focus here is on a non-trivially fast exponential time
  algorithm to ``approximately" solve $0,1$-integer linear programs with variables taking values from a bigger set $E$.
  \fi

  Very recently, \cite{DBLP:journals/corr/abs-1801-09488} constructed the first Boolean CSP which has a quantiative lower bound (assuming SETH) and has a nontrivial (e.g., not $2^{n-o(n)}$) upper bound. The constraints of the CSP do not have a finite description but are rather all possible constraints which have a prescribed partial polymorphism known as a ``2-edge operator.'' Their upper bound algorithm involves a ``meet in the middle strategy'' to get a $2^{n/2}$ runtime. A related problem they explore (corresponding to ``$k$-near-unanimity'' operators) gives an nontrivial upper bound involving a variant of Sch{\"o}ning's algorithm, although the correctness of the algorithm is conditioned on the sunflower conjecture.

  These methods also connect to the $(2+\epsilon)$-SAT problem of \cite{DBLP:journals/siamcomp/AustrinGH17}. Succinctly, the $(2+\epsilon)$-SAT problem, for $\epsilon = 1/k$, is the following: given an instance of $(2k+1)$-SAT in which there exists an assignment which satisfies at least $k$ literals in every clause, find a solution to the ordinary $k$-SAT instance.  For $E = [0, \frac{k}{2k+1}) \cup (\frac{k+1}{2k+1}, 1]$, there is a reduction of $(2+1/k)$-SAT to finding a solution belonging to $E^n$ for the basic LP of the $(2k+1)$-SAT instance. Note that this implies the existence of a $\left(1 + \frac{1}{2k+1}\right)^n\poly(n)$ time algorithm for $(2+1/k)$-SAT. Since $(2+\epsilon)$-SAT was shown to be NP-hard for every $\epsilon > 0$~\cite{DBLP:journals/siamcomp/AustrinGH17}, this reduction also shows that our LP problem is NP-hard for $E = [0, \alpha] \cup [1-\alpha, 1]$ for each fixed $\alpha < \frac{1}{2}$.

\subsection{The road ahead}
This algorithm for linear programming spurs many questions for further investigation. The following are a sample of potential directions of exploration.

\begin{itemize}
\item For succinctness, this article only describes testing feasibility of a linear programs. One may also consider the optimization version of the question where one seeks to maximize $c^Tx$ subject to $Ax \le b$ and $x \in \{0, 1\}^n$. Assume that $M$ is the optimal value. One way to phrase such an inquiry is to desire to find $x \in E^n$ such that $Ax \le b$ and $c^Tx \ge M$. This can be achieved with essentially the same complexity by performing binary search on candidate values $M'$ and adding $c^Tx \ge M'$ to the linear program. 

\item Although the focus of this paper is optimizing in powers of subsets $E \subset [0, 1]$, one can also investigate product sets $E_1 \times E_2 \times \cdots \times E_n$. As the relevant potentials are multiplicative, the analysis of this paper can generalize in a straightforward manner, as long as we have a uniform bound on the traction $\tau(E_i)$ (see \ref{eq:traction}).

  Another extension would be product sets $E_1 \times \cdots \times E_n$, where each $E_i \subset [0, 1]^k$ for some constant $k$. Such reductions are relevant in reducing from non-Boolean CSPs.

\item Also for succinctness, we restrict the exposition to having the bridging between $\{0, 1\}^n$ and $[0, 1]^n$ to be a product set $E^n$, but our method seems to have the capability, with several additional technical ideas, to generalize to more complex regions $R \subseteq [0, 1]^n$. An interesting example is
  \[
    R = \{x \in [0, 1]^n \mid \exists y \in \{0, 1\}^n , \ \|x - y\| \le \epsilon\},
  \]
  where $\epsilon > 0$ and $\|\cdot\|$ is any norm on $[0, 1]^n$. Note that $R = E^n$ with $E = [0, \epsilon] \cup [1-\epsilon, 1]$ covers the case $\|\cdot\|$ is the $L^{\infty}$ norm.
  
\item Another important question is finding a derandomization of the random walk algorithm. Moser and Scheder~\cite{MoserScheder2011} did successfully derandomize Sch\"{o}ning's algorithm, so we envision that similar methods should be able to derandomize our algorithm.

\item This article shows how this LP algorithm can give exponential time random walk algorithms for a variety of Boolean CSPs and Promise CSPs where the ``promise domain'' is Boolean. It is possible to extend to higher domains by noting that any CSP can be expressed a $\{0,1\}$-integer program by having indicator variables $x_{i,j}$ which represent if the $i$th variable in the CSP is equal to $j$.
\end{itemize}

\label{subsec:org}
\medskip\noindent \textbf{Paper Organization.} In Section~\ref{sec:IP}, we formally state the main result as well as sketch the algorithm. In Section~\ref{sec:analysis}, we prove that this algorithm is correct and has the claimed run-time. In Section~\ref{sec:PCSP}, we prove some applications of the main result, including recovery of Sch\"{o}ning's random walk algorithm. In Appendix~\ref{app:omit}, we include the proofs (mainly some calculations) omitted in the body of the paper. 

\section{Random walk algorithm}\label{sec:IP}

Consider any linear program $Ax \le b$, where $A \in \mathbb Q^{m \times n}, b \in \mathbb Q^n$. Assume that $m$ is bounded by a subexponential function of $n$.  Treat this as a $0$-$1$ integer program, so we desire $x \in \{0, 1\}^n$. Unless one refutes the Strong Exponential Time Hypothesis (SETH), determining such an $x$ requires at least $2^{(1-o(1))n}$ time. That said, we show in this section, that one can get significantly better runtimes if one accepts an approximate solution. We now define what we mean by ``approximate.''

Let $E \subset [0, 1]$ be a closed subset with $\{0, 1\} \subset E$ which is the union of disjoint intervals of nonzero length. As stated below, these intervals of $E$ are explicitly given as part of the problem statement. The set $E$ tracks which errors are allowed. That is, we define an approximate solution to be any $x \in E^n$ which satisfies the linear program. This leads to the following theorem.

\begin{thm}\label{thm:IP}
  Let $Ax \le b$ be a linear program with $A \in \mathbb Q^{m \times n}, b \in \mathbb Q^n$. Let $E = [c_1, d_1] \cup [c_2, d_2] \cup \cdots \cup [c_k, d_k]$ be a sequence of intervals with rational endpoints with
  \[
    0 = c_1 < d_1 < c_2 < d_2 < \cdots < c_k < d_k = 1.
  \]
  Assume the promise that there exists $x \in \{0, 1\}^n$ which satisfies the linear program. There exists a randomized algorithm computing a solution $x \in E^n$ to the linear program in time
  \[
    c_E(2 - \meas(E))^{n}\poly(n+m),
  \]
  where $\meas(E)$ is the sum of the lengths of the intervals of $E$, and $c_E$ depends only on $E$.
\end{thm}

At a high level, think of the problem as a two-player game. Alice has in mind some solution $x \in \{0, 1\}^n$ to $Ax \le b$ while Bob is trying to find $y \in E^n$ such that $Ay \le b$ and $y$ is as similar to $x$ as possible. To keep track of this similarity, we have a ``score'' (potential function) for the game which states how similar Bob's $y$ is to $x$. At first Bob's guess, say $y^{(0)}$ does not satisfy $Ay^{(0)} \le b$, but over time Bob can use the score as a heuristic to guide a random walk so that he gets close enough in similarity to $x$ to solve $Ay \le b$.

\subsection{Potential Function}\label{subsec:IP-potential}

Consider a game in which Alice has in mind $a \in \{0, 1\}$ and Bob guesses an interval $[c_b, d_b]$ in $E = [c_1, d_1] \cup \cdots \cup [c_k, d_k]$ with $c_1 = 0$ and $d_k = 1$. We give Bob a score $U_a : [k] \to (0, 1]$ satisfying the following rules
\begin{itemize}
\item Correct guesses have full potential $U_0(1) = 1$ and $U_1(k) = 1$.
\item The potential changes monotonically, for any $i \in [k-1]$ we have
  \begin{align*}
    \frac{U_{0}(i+1)}{U_0(i)} &= \frac{1-c_{i+1}}{1 - d_i}\\
    \frac{U_{1}(i)}{U_{1}(i+1)} &= \frac{d_i}{c_{i+1}}
  \end{align*}
\end{itemize}

The precise function which satisfies this condition is as follows:
\begin{align}
  U_0^E(i) &=
    \prod_{j=1}^{i-1} \frac{1-c_{j+1}}{1-d_{j}} \label{eq:0}\\
  U_1^E(b) &= 
    \prod_{j=1}^{k - i} \frac{d_{k+1-j}}{c_{k-j}} \label{eq:1},
\end{align}
where the empty product is equal to $1$.

The potential functions are chosen so that during each step of the random walk, the expected value of the potential stays the same over time. The quantitative justification is in Claim~\ref{claim:potential}.

Now, in the actual game, Alice has in mind an $n$-bit vector $x \in \{0, 1\}^n$ (the 0-1 integer program solution) and Bob has $n$ interval guesses $y \in [k]^n$. The potential in this case is just the product of the coordinate-wise potentials.

\[
  U_x(y) := \prod_{i=1}^n U_{x_i}(y_i).
\]

Note that we still have that $U_x(y) \in (0, 1]$ for all $x \in \{0, 1\}^n$ and $y \in [k]^n$.

In order to make quantitative guarantees about the runtime of our algorithm, we need some parameters which quantify how ``well-conditioned'' our potential function is. Define the \emph{quanta} of $E$ to be
\[
  \gamma(E) = \min (\{U_0(1), U_1(k)\}.
\]
In other words, $\gamma(E)$ captures how small the potential of one coordinate can be.

Define the \emph{traction} $\tau(E)$ of $E$ to be
\begin{align}
  \tau(E) := \min_{\substack{a \in \{0, 1\}\\i,j \in [k], i \neq j\\U_a(i)\ge U_a(j) \neq 0}} \frac{U_a(i)}{U_a(j)}.\label{eq:traction}
\end{align}

Roughly speaking, $\tau(E)$ measures the minimum amount the potential can go up by when it increases. The latter is quite important in order to ensure that our potential does not get ``stuck'' by increasing only a negligible amount on each step.

\subsection{Sketch of Random Walk Algorithm}

With the potential function defined, we can now describe the random walk algorithm at a high level.

  \begin{framed}
  	
  	\begin{center}
  		\textbf{Algorithm \thesection.1:} A randomized algorithm for solving approximate $\{0, 1\}$ integer programs.
  	\end{center}
  
    \textbf{Input:} $Ax \le b$, $A \in \mathbb Q^{m \times n}, b \in \mathbb Q^m$. $E = [c_1, d_1] \cup \cdots \cup [c_k, d_k]$, $0 = c_1 < d_1 < c_2 < d_2 < \cdots < c_k < d_k = 1, c_i, d_i \in \mathbb Q$.
    
    The algorithm uses a parameter $\epsilon > 0$ governing the number of iterations.
    \begin{itemize}
    \item[1.] For each $i \in \{1, \hdots, n\}$. Sample $\sigma^1_i \in \{1, \hdots, k\}$ according to the distribution.
      \begin{align*}
        \Pr[\sigma^1_i = j] = \frac{c_{j+1} - d_{j-1}}{2 - \meas(E)}.
      \end{align*}
      where $d_0 = 0$ and $c_{k+1} = 1$.
    \item[2.] For $t$ in $\left\{1, \hdots, T_{n,E}\right\}$
      \begin{itemize}
      \item[3.] Check if there exists $y \in \prod_{i=1}^n [c_{\sigma^t_i}, d_{\sigma^t_i}]$ such that $Ay \le b$ If so, \textbf{Output $y$}.
      \item[4.] Find $p_{i,j}, i \in [n], j \in [k]$, the probability that $\sigma^t_i$ changes to $j$, by solving the following feasibility problem.
        \begin{align}
          \sum_{i=1}^n \sum_{j=1}^k p_{ij} &= 1, p_{i,j} \ge 0\nonumber\\
          \forall i \in \{1, \hdots, n\}, p_{i\sigma^t_i} &= 0\nonumber\\
          \forall x \in [0, 1]^n \text{ such that }Ax \le b, &\sum_{i=1}^n \sum_{j=1}^k p_{ij} \left[(1 - x_i)\frac{U_0(j)}{U_0(\sigma^t_i)} + x_i\frac{U_1(j)}{U_1(\sigma^t_i)}\right] \ge 1 \label{eq:key}
        \end{align}
      \item[5.] Sample $(i, j) \sim [n] \times [k]$ according to the probability $p_{ij}$. Set $\sigma^{t+1} = \sigma^t$ except $\sigma^{t+1}_i = j$.
      \end{itemize}
    \item[6.] \textbf{Output} ``Fail.'' 
    \end{itemize}
  
\end{framed}

We claim the following

\begin{thm}\label{thm:random-walk2}
  Assume that there exist $x \in \{0, 1\}^n$ such that $Ax \le b$. Then, if Bob starts his random walk at $\sigma \in \{1, \hdots, k\}^n$, he will find a solution to $Ay \le b$ with $y \in E^n$ with probability at least
  \[
    \frac{1}{2}U_x(\sigma)^{1+1/n}
  \]
  as long as
  \[
    T_{n, E} = \left\lceil\dfrac{(n+1)\log(1/\gamma(E))}{\log\left[1 + \frac{1}{2n}\left(1 - \frac{1}{\tau(E)}\right)^2\right]} + 2\right\rceil \ .
  \]
  Furthermore, computing the steps of this random walk can be done in polynomial time.
\end{thm}

\section{Analysis of Random Walk Algorithm}\label{sec:analysis}

The proof of Theorem~\ref{thm:random-walk2} consists of three parts.

\begin{enumerate}
\item Prove that Algorithm 2.1 is well-defined. That is, the probability distribution $p_{i,j}$ described in Step 4 exists.
\item Prove that Algorithm 2.1 can be implemented in randomized polynomial time. That is, the probability distribution $p_{i,j}$ described in Step 4 can be computed in polynomial time.
\item Prove that the probability of success of Algorithm 2.1 is correctly stated.
\end{enumerate}

The following three subsections establish each of these parts.

\subsection{Algorithm 2.1 is well-defined}

In this subsection, we show that the probability distribution $p_{i,j}$ described in Step 4 exists. To do this, we use a special case of the minimax theorem~\cite{v.Neumann1928}.

\begin{thm}[Minimax Theorem]
  Let $K_1 \subset \mathbb R^M, K_2 \subset \mathbb R^N$ be compact, convex sets. Let $f : K_1 \times K_2 \to \mathbb R$ be a function which is affine in both coordinates. Then,
  \[
    \min_{x \in K_1} \max_{y \in K_2} f(x, y) = \max_{y \in K_2} \min_{x \in K_1} f(x, y).
  \]
\end{thm}
\
To apply this, let

\begin{align*}
  K_1 &= \{x \in [0, 1]^n : Ax \le b\}\\
  K_2 &= \left\{p \in [0, 1]^{n \times k} \middle| \sum_{i=1}^n\sum_{j=1}^k p_{i,j} = 1, \forall i, \forall j, p_{i,j} \ge 0, p_{i\sigma_i^t} = 0 \right\}
\end{align*}
That is, $K_1$ is in some sense the convex hull of Alice's strategies, and $K_2$ is the distribution of Bob's potential modifications to $\sigma_i^t$.

We then define $f(x, p)$ to be precisely the functional which appears in Step 4 of Algorithm 2.1
\[
  f(x, p) := \sum_{i=1}^n \sum_{j=1}^k p_{ij} \left[(1 - x_i)\frac{U_0(j)}{U_0(\sigma^t_i)} + x_i\frac{U_1(j)}{U_1(\sigma^t_i)}\right].
\]
Note that this expression is affine in $p$ and is also affine in $x$ (the other terms are constant within Step 4 of the algorithm). Note that $p$ exists as expected in Algorithm 2.4 if and only if
\[
  \max_{p \in K_2}\min_{x \in K_1} f(x, p) \ge 1.
\]
That is, at each step, Bob has a strategy that works universally over all possible strategies for Alice. 
By the Minimax Theorem, it suffices to show that $\min_{x \in K_1} \max_{p \in K_2} f(x, p) \ge 1.$ That is, it suffices to prove the following.
\begin{claim}\label{claim:potential}
  For all $x \in K_1$ there exists $p \in K_2$ such that $f(x, p) \ge 1$.
\end{claim}

\begin{proof}
  Fix this $x \in K_1$. We'll show that there exists $p \in K_2$ which has exactly one nonzero element, and this element is equal to $1$.

  Since in Step 3, the algorithm verified that $\prod_{i=1}^n [c_{\sigma_i^t}, d_{\sigma_i^t}]$ is disjoint from $K_1$, we know that there exists $i \in [n]$ such that $x_i \not\in [c_{\sigma_i^t}, d_{\sigma_i^t}]$. For notational simplicity, let $j = \sigma_i^t$. Thus, either $x_i < c_j$ or $x_i > d_j$.

  If $x_i < c_j$, then $j \ge 2$ as $c_j > 0$ but $c_1 = 0$. Thus, consider $p \in K_2$ such that $p_{i, j-1} = 1$ and $p_{i', j'} = 0$ otherwise. Now we have that
  \begin{align*}
    f(x, p) &= (1 - x_i)\frac{U_0(j-1)}{U_0(j)} + x_i \frac{U_1(j-1)}{U_1(j)}\\
            &= (1 - x_i)\frac{1 - d_{j-1}}{1 - c_j} + x_i \frac{d_{j-1}}{c_j}\\
            &= \frac{c_j - (c_j - d_{j-1})x_i - c_jd_{j-1}}{(1-c_j)c_j}\\
            &\ge \frac{c_j - (c_j - d_{j-1})c_j - c_jd_{j-1}}{(1-c_j)c_j}\text{ (as $x_i \le c_j$)}\\
            &= 1.
  \end{align*}
  Note that $c_j < 1$ so we are not dividing by $0$.

  The case $x_i > d_j$ is handled by having $p_{i, j+1} = 1$ and $p_{i', j'} = 0$ otherwise. The analysis is analogous.
\end{proof}

Thus, the distribution $p_{i,j}$ needed by Step 4 does indeed exist.

\subsection{Algorithm 2.1 has polynomial time complexity}

Next, we show that Algorithm 2.1 can be implemented to run in (randomized) polynomial time. It is obvious that Steps 1, 5, 6 run in polynomial time. The number of possible values of $t$ in Step 2 is polynomial time as $\frac{c_1n}{\log(1 + c_2/n)} = O_E(n^2)$. Step 3 runs in polynomial time as it is equivalent to checking if a linear program is feasible. The non-trivial step to justify
is Step 4.

Let $P \subset [0, 1]^{n \times k}$ be the region of $p$ which satisfy the conditions stipulated in Step 4. It is clear that $P$ is convex, and we showed that it is non-empty in the previous subsection. Since the set $K_1 = \{x \in [0, 1]^n : Ax\le b\}$ is a compact, convex polytope, it suffices to check the condition (\ref{eq:key}) for $x$ on the vertices of this polytope.
By standard results in linear programming, the ``representation complexity'' of the vertices of $K$ is some polynomial in the representation complexity of the system $Ax \le b, x \in [0, 1]$. In particular, this implies that $P$ is itself a polytope as its finitely many facets are described by (\ref{eq:key}) and these have bounded complexity as we can restrict $x$ to the vertices of $K$. Thus, \cite{grotschel2012geometric} describes $P$ as a ``well-defined'' polytope. By Theorem~6.4.1 of \cite{grotschel2012geometric}, finding a point $p \in P$ (or asserting that none exists), can be done in polynomial time as long as there exists a strong separation oracle that itself can be computed in polynomial time.\footnote{See also \url{https://www.cs.cmu.edu/afs/cs.cmu.edu/academic/class/15859-f11/www/notes/lecture09.pdf}.} In other words, for any $p \not \in P$, one need to be able to efficiently compute a hyperplane $H$ such that $P$ is strictly on one side and $p$ is strictly on the other side.

To go about doing this, we first compute the $x \in K$ which $p \not\in P$ performs the worst against. We can do this by solving the following linear program.
\begin{framed}
  \begin{align*}
    \textbf{Given:}&&p \not\in P\\
                      \text{minimize}&&M=\sum_{i=1}^n \sum_{j=1}^k p_{ij} \left[(1-x_i)\frac{U_0(j)}{U_0(\sigma_i^t)} + x_i \frac{U_1(j)}{U_1(\sigma_i^t)}\right]\\
    \text{subject to}&&Ax \le b.
  \end{align*}
\end{framed}
Since $p \not\in P$, we know that $M < 1$. Thus the hyperplane
\[
  H := \left\{p \in \mathbb R^{n\times k} \middle| \sum_{i=1}^n \sum_{j=1}^k p_{ij} \left[(1-x_i)\frac{U_0(j)}{U_0(\sigma_i^t)} + x_i \frac{U_1(j)}{U_1(\sigma_i^t)}\right] = \frac{M+1}{2}\right\}
\]
can be efficiently computed and strictly separates $p$ and $P$. Thus, the random walk itself can be performed efficiently.

\subsection{Algorithm 2.1 succeeds with decent probability}

We need the following important claim about random walks.
\begin{lem}\label{lem:random}
  Let $(X_n \sim [0, 1])_{n=1}^{T}$ be a sequence of random variables and $\tau > 1$ such that
  \begin{itemize}
  \item If $X_t \in \{0, 1\}$ then $X_{t+1} = X_t.$
  \item If $X_t \in (0, 1)$ then $\frac{X_{t+1}}{X_t} \not \in (1/\tau, \tau)$ surely.
  \item The sequence is a submartingale: $\mathbb E[X_{t+1} | X_t, \hdots, X_1] \ge X_t.$
  \end{itemize}

  Let $Q_t = \Pr[X_t = 1]$. Then, for all $\epsilon > 0$ whenever
  \begin{align}
    T \ge \frac{\log(1/\mathbb E[X_1^{1+\epsilon}])}{\log\left(1 + \frac{\epsilon}{2}\left(1 - \frac{1}{\tau}\right)^2\right)} + 2 \label{eq:T}
  \end{align}
  we have that
  \[
    Q_T \ge \frac{1}{2}\mathbb E[X_1^{1+\epsilon}].
  \]
\end{lem}

\begin{proof}
  Fix $\epsilon > 0$. Define $Y_t = X_t^{1+\epsilon}$. We claim that $Y_t$ performs even better than a submartingale. We prove this in Appendix~\ref{app:omit}.
  \begin{claim}\label{claim:submart}
    \[
      \mathbb E[Y_{t+1} | Y_t, \hdots, Y_1 \wedge Y_t \neq 1] \ge (1+\delta)Y_t,
    \]
    where
    \[
      \delta = \frac{\epsilon}{2}\left(1 - \frac{1}{\tau}\right)^2 > 0.
    \]
  \end{claim}

  Now, observe that $Y_t = 1$ iff $X_t = 1$. Thus, by taking the expectation of the inequality in the claim over all $Y_t, \hdots, Y_1$ (subject to $Y_t \neq 1$), we get
  \[
    \mathbb E[Y_{t+1} | Y_t \neq 1] \ge (1+\delta) \mathbb E[Y_t | Y_t\neq 1].
  \]
  Therefore,
  \begin{align*}
    \mathbb E[Y_{t+1}] &= Q_t\mathbb E[Y_{t+1} | Y_t = 1] + (1 - Q_t)\mathbb E[Y_{t+1} | Y_t \neq 1]\\
                       &\ge Q_t + (1 - Q_t)(1+\delta)\mathbb E[Y_t | Y_t \neq 1]\\
                       &= Q_t + (1 - Q_t)(1+\delta)\frac{\mathbb E[Y_t] - Q_t}{1 - Q_t}\\
                       &= Q_t + (1+\delta)(\mathbb E[Y_t] - Q_t)\\
                       &= (1+\delta)\mathbb E[Y_t] - \delta Q_t
  \end{align*}
  Fix $Q = \mathbb E[Y_1]$. If $Q_T  > Q/2$, we are done. Otherwise, for all $t \in \{1, \hdots, T-1\}$ we have that $Q_t \le Q/2$ so
  \[
    \mathbb E[Y_{t+1}] \ge (1+\delta)\mathbb E[Y_t] - \delta \frac{Q}{2}.
  \]
  Thus,
  \[
    \mathbb E[Y_{t+1} - Q/2] \ge (1 + \delta)\mathbb E[Y_t - Q/2].
  \]
  Thus by induction \[\mathbb E[Y_T - Q/2] \ge (1+\delta)^{T-1} \mathbb E[Y_1 - Q/2] = (1+\delta)^{T-1} \frac{Q}{2}.\]
  Since clearly $1 > \mathbb E[Y_T - Q/2]$. We must have that
  \[
    \frac{\log \frac{2}{Q}}{\log(1+\delta)} + 1 > T,
  \]
  a contradiction. This concludes the proof.
\end{proof}

Recall that $x \in \{0, 1\}^n$ is Alice's solution to the integer program $Ax \le b$. When Bob does his random walk, he starts with some $\sigma^1 \in \{1, \hdots, k\}^n$ and computes random variables $\sigma^2, \sigma^3, \hdots, \sigma^T \sim \{1, \hdots, k\}^n$.If Bob succeeds at time step $t$, then assume that $\sigma^{s} = \sigma^t$ for all $s \ge t$. Step~4 of Algorithm~2.1 ensures that
\[
  \mathbb E\left[\frac{U_x(\sigma^{t+1})}{U_x(\sigma^t)} \middle| \sigma^t, \hdots, \sigma^1\right] = \sum_{i=1}^n \Pr[\sigma^{t+1}_i \neq \sigma^t_i]\mathbb E\left[\frac{U_x(\sigma^{t+1}_i)}{U_x(\sigma^t_i)} \mid \sigma^{t+1}_i \neq \sigma^t_i\right] \ge 1,
\]
so
\[
  \mathbb E[U_{x}(\sigma^{t+1}) | \sigma^t, \hdots, \sigma^1] \ge U_{x}(\sigma^t).
\]

Now define the random variable
\[
  X_t = \begin{cases}
    1 & \text{$\sigma^t$ yields a solution}\\
    U_x(\sigma^t) & \text{otherwise.}
  \end{cases}
\]
In particular, $X_t = 1$ if and only if Bob ``wins'' on step $t$. This sequence $(X_n)_{n=1}^{\infty}$ (with $\tau$ being the ``traction'' $\tau(E)$) satisfied the conditions of Lemma~\ref{lem:random}with $\epsilon = 1/n$. Note that $\mathbb E[X_1^{1+1/n}]$ is at least $\mathbb E[U_x(\sigma^t)^{1+1/n}] \ge \gamma(E)^{n+1}$. 

In particular $T = T_{n, E}$ satisfies (\ref{eq:T}). Therefore, Bob will have ``won'' by step $T$ with probability at least $\frac{1}{2}U_x(\sigma^1)^{1+1/n},$ as desired.

This concludes the proof of Theorem~\ref{thm:random-walk2}. Now, we show that this implies Theorem~\ref{thm:IP}.

\subsection{Proof of Theorem~\ref{thm:IP}}\label{subsec:proof}

Consider a distribution $q \sim [k]$ which maximizes the following quantity.
\begin{align}
  \beta(E) = \max_{q} \min(\mathbb E_{i \sim q}[U_0^E(i)], \mathbb E_{i\sim q}[U_1^E(i)]) \label{eq:beta}.
\end{align}

Recall that $q$ is the starting distribution for each coordinate in our algorithm (and different coordinates are sampled independently according to $q$). The expression $\beta(E)$ captures the initial potential in that coordinate, which by Lemma~\ref{lem:random} is tied to the ultimate success probability.

Note that $q$ and $\beta$ can be computed in $\poly(k)$ time using a simple linear program. By Theorem~\ref{thm:random-walk2}, we have that if each coordinate of $\sigma^1$ is sampled independently from the distribution $q$, then the probability of success, in terms of some integral solution $x \in \{0, 1\}^n$ to $Ax \le b$, is at least 
\begin{align*}
  \frac{1}{2}\mathbb E[U_x(\sigma^1)^{1+1/n}] &\ge \frac{1}{2}\mathbb E[U_x(\sigma^1)]^{1+1/n}\\
                                                   &= \frac{1}{2}\left[\prod_{i=1}^n \mathbb E[U_{x_i}(\sigma^1_i)]\right]^{1+1/n}\\
                                                   &\ge \frac{1}{2}\beta(E)^{n+1}.                                               
\end{align*}

Now consider the following bounded on $\beta$ whose proof is in Appendix~\ref{app:omit}. It says that $\beta$ can be larger than $1/2$ (and thus the success probability better than random guessing) when $E$ has positive measure.

\begin{claim}\label{claim:beta}
  For all $E = [c_1, d_1] \cup \cdots \cup [c_k, d_k]$ with $0 =c_1 \le d_1 < c_2 \le d_2 < \cdots < c_k \le d_k = 1$,
  \[\beta(E) \ge \frac{1}{2 - \meas(E)},\]
  where $\meas(E) = \sum_{i=1}^k (d_i - c_i)$.
\end{claim}

Thus, each run of Algorithm~2.1 runs in polynomial time at succeeds with probability at least $\frac{1}{2}(2-\meas(E))^{-n-1}$. Thus, after $n(2-\meas(E))^{n}$ iterations, the algorithm will succeed with high probability. Since each iteration runs for $T_{n,E}\poly(n, m, k) = O_E(\poly(n, m))$ steps, we have the desired runtime of $O_E((2-\meas(E))^n\poly(n, m))$. This completes the proof.

\begin{rem}
  Although this bound for $\beta$ is tight in certain cases, such as $E$ is the union of two intervals. In other cases, it is not. Consider $E = [0, 1/5] \cup [2/5, 3/5] \cup [4/5, 5/5]$. Then, the Claim shows that $\beta(E) \ge \frac{5}{7}$. But, if one considers the distribution which always picks the middle interval, then note that since
  \[
    U_0(2) = U_1(2) = \frac{1 - 2/5}{1 - 1/5} = \frac{3}{4},
  \]
  then one has an improved bound that $\beta(E) \ge \frac{3}{4}$.
\end{rem}

\section{Applications to CSPs}\label{sec:PCSP}

In this section, we show that Theorem~\ref{thm:IP} can be applied to give fast exponential time algorithms for many Constraint Satisfaction Problems (CSPs).

\subsection{CSPs and Partial Polymorphisms}\label{subsec:PCSP-prelim}

Recall from the CSP literature that a Boolean constraint can be viewed as $R \subset \{0, 1\}^{\ar}$, where $\ar$ is known as the ``arity'' of the constraint. A \emph{template} is a set $\Gamma$ of these constraints. $\CSP(\Gamma)$ is the decision problem of deciding if a CNF with constraints from $\Gamma$ is satisfiable.

Consider a positive integer $L$ and a function $f : \{0, 1\}^L \to \{0, 1, \perp\}$ such that for every $R \in \Gamma$ where $R \subset \{0, 1\}^{\ar}$ and every sequence $x^1, \hdots, x^L \in R$ satisfying 
\[
  y_i := f(x^1_i, \hdots, x^L_i) \neq \perp \quad \text{for $i=1,2,\dots,\ar$} \ ,
\]
it holds that $(y_1, \hdots, y_{\ar}) \in R$. Such $f$ are precisely the \emph{partial polymorphisms} of $f$. Note that these differ from normal polymorphisms because we are allowed to output $\bot$.

As an example, for $k$-SAT, a suitable template is $\Gamma = \{R \subset \{0, 1\}^k : |R| = 2^k - 1\},$ and one can verify that for all $L \ge 1$ not divisible by $k$, the following is a partial polymorphism
\[
  f_L(x) = \begin{cases}
    0 & \Ham(x) < \frac{L}{k}\\
    1 & \Ham(x) > L - \frac{L}{k}\\
    \perp & \text{otherwise}.
    \end{cases}
\]

This family of partial polymorphisms actually falls under a whole family known as \emph{threshold partial polymorphisms}
\begin{df}
  Let $E = [c_1, d_1] \cup \cdots \cup [c_k, d_k]$ be a union of a sequence of intervals such that $c_i, d_i \in \mathbb Q$ and
  \[
    0 = c_1 < d_1 < c_2 < d_2 < \cdots < c_k < d_k = 1.
  \]
  Let $\eta : E \to \{0, 1\}$ such that $\eta(0) = 0, \eta(1) = 1$ and $\eta$ is constant within each subinterval. Let $L$ be a positive integer such that $Lc_i$ is non-integral for $c_i > 0$ and $Ld_i$ is non-integral for $d_i < 1$. We say that $f : \{0, 1\}^L \to \{0, 1, \perp\}$ \emph{$(E, \eta)$-threshold partial function} if
  \[
    f(x) = \begin{cases}
      \eta\left(\frac{\Ham(x)}{L}\right) & \frac{\Ham(x)}{L} \in E\\
      \perp & \text{otherwise}.
    \end{cases}
  \]
\end{df}
In particular, $k$-SAT has a threshold partial function with $E = [0, 1/k] \cup [1-1/k, 1]$ and $\eta(x) = 0$ if $x \le 1/k$ and $\eta(x) = 1$ if $x \ge 1 - 1/k$.

\begin{thm}
  Let $\Gamma$ be a Boolean CSP template and let $E = [c_1, d_1] \cup \cdots \cup [c_k, d_k]$ and $\eta : E \to \{0, 1\}$ be such that for infinitely many positive integers $L$, $\Gamma$ has a partial polymorphism $f : \{0, 1\}^L \to \{0, 1, \perp\}$ which is an $(E, \eta)$-threshold partial function. Then, $\CSP(\Gamma)$ can be solved in randomized $O^*_E((2 - \meas(E))^n)$ time, where $n$ is the number of variables.
\end{thm}

\begin{proof}
  Given an instance of $\CSP(\Gamma)$ let $Ax \le b$ be its integer programming relaxation where $x \in \{0, 1\}^n$ (e.g., see the Basic LP reduction of \cite{BrakensiekGuruswami2019}). Use Theorem~\ref{thm:IP} to find $y \in E^n_{\delta}$ which also satisfies $Ay \le b$, where we define
  \[
    E_{\delta} = [0, d_1 - \delta] \cup [c_1 + \delta, d_2 - \delta] \cup \cdots \cup [c_k + \delta,  1].
  \]
  This takes $O^{*}_{E_{\delta}}((2 - \meas(E) + 2k\delta)^{n})$ time\footnote{Note the change in the hidden constant depending on $E_{\delta}$ is a function of $\gamma(E_{\delta})$, $\tau(E_{\delta})$, and $\meas(E_{\delta}$, all of which change by a negligible amount if $\delta$ is small.}, which is at most $O^*_E((2 - \meas(E))^{n+1})$ if we take $\delta = O(\frac{1}{kn})$. 

  We claim that $(\eta(y_1), \hdots, \eta(y_n))$ satisfies the original instance of the $\CSP(\Gamma)$. The proof is essentially the same as that of Theorem~4.1 of \cite{BrakensiekGuruswami2019}. In essence, for each clause $(x_{i_1}, \hdots, x_{i_{\ar}}) \in R$, the fact that $(y_1, \hdots, y_n)$ solves the LP relaxation means that there exist assignments $z^1, \hdots, z^r \in R$ and weights  $w_1, \hdots, w_r \in [0, 1]$such that
  \[
    \sum_{i=1}^r w_i = 1 \text{ and } \sum_{i=1}^rw_iz^i = (y_{i_1}, \hdots, y_{i_{\ar}}).
  \]
  Consider $L > n\delta$ such that $f : \{0, 1\}^L \to \{0, 1, \perp\}$ is an $(E, \eta)$-threshold partial function. Then as in \cite{BrakensiekGuruswami2019}, pick weights $W^1, \hdots, W^r$ such that $|W^i - Lw_i| \le 1$ so then
  \[
    \sum_{i=1}^r W_iz^i =: (Y_{i_1}, \hdots, Y_{i_{\ar}}) \approx (Ly_{i_1}, \hdots, Ly_{i_{\ar}}).
  \]
  Because $y_{i_j} \in [c_{\ell}+\delta, d_{\ell}-\delta]$ or $[0, d_1-\delta]$ or $[c_k+\delta, 1]$, we have that $\frac{Y_{i_j}}{L} \in E$ for all $i$ and $\frac{Y_{i_j}}{L}$ will be in the same interval of $E$ as $y_{i_j}$. By plugging in $z^1, \hdots, z^r$ into $f$ with weights $W^1, \hdots, W^r$, we then get that $(\eta(Y_{i_1}/L), \hdots, \eta(Y_{i_r}/L)) = (\eta(y_{i_1}), \hdots, \eta(y_{i_r}))\in R$ which shows that $(\eta(y_1), \hdots, \eta(y_n))$ is a satisfying assignment.
\end{proof}

\appendix

\section{Omitted Proofs}\label{app:omit}

\subsection{Proof of Claim~\ref{claim:submart}}

First, we need to prove a technical inequality.

\begin{prop}\label{prop:calc}
  For all $\tau \ge 1$ and $\epsilon > 0$, we have that
  \[
    \min\left(\frac{1}{\tau^{1+\epsilon}} - \frac{1+\epsilon}{\tau} + \epsilon, \tau^{1+\epsilon} - (1+\epsilon)\tau + \epsilon\right) \ge \delta  = \frac{\epsilon}{2}\left(1 - \frac{1}{\tau}\right)^2.
  \]
\end{prop}

\begin{proof}
  This requires proving two separate inequalities.
  \begin{align}
    \frac{1}{\tau^{1+\epsilon}} - \frac{1+\epsilon}{\tau} + \epsilon &\ge \frac{\epsilon}{2}\left(1 - \frac{1}{\tau}\right)^2 \label{eq:1/tau}\\
    \tau^{1+\epsilon} - (1+\epsilon)\tau + \epsilon &\ge \frac{\epsilon}{2}\left(1 - \frac{1}{\tau}\right)^2.\label{eq:tau}
  \end{align}
  By multiplying both equations by $\tau^2$ and moving all the terms to the right side, it suffices to verify that the following polynomial-like expressions are nonnegative for $\tau \in [0, 1)$ (for any fixed $\epsilon \ge 0$).
  \begin{align*}
    g(\tau) &:= \frac{\epsilon}{2}\tau^2 - \tau + \tau^{1-\epsilon} - \frac{\epsilon}{2}\\
    h(\tau) &:= \tau^{3+\epsilon} - (1+\epsilon)\tau^3 + \frac{\epsilon}{2}\tau^2 + \epsilon \tau - \frac{\epsilon}{2}.
  \end{align*}
  Note that $g(1) = h(1) = 0$, so it suffices to show that $g'(\tau), h'(\tau) \ge 0$ for all $\tau \in [1, \infty)$.

  Note that $g'(\tau) = \epsilon \tau + (1-\epsilon)\tau^{-\epsilon} - 1$, but by Jensen's inequality applied to the convex function $\tau^x$, we have that $\epsilon \tau + (1-\epsilon)\tau^{-\epsilon} \ge \tau^{\epsilon + (1-\epsilon)(-\epsilon)} = \tau^{\epsilon^2} \ge 1$. Thus, $g'(\tau) \ge 0$ for all $\tau \in [1, \infty)$, so (\ref{eq:1/tau}) follows.

  For $h$, observe that
  \begin{align*}
    h'(\tau) &= (3+\epsilon)\tau^{2+\epsilon} - 3(1+\epsilon)\tau^2 + \epsilon \tau + \epsilon\\
    h''(\tau) &= (3+\epsilon)(2+\epsilon)\tau^{1+\epsilon} - 6(1+\epsilon)\tau + \epsilon\\
    h'''(\tau) &= (3+\epsilon)(2+\epsilon)(1+\epsilon) \tau^{\epsilon} - 6(1+\epsilon).
  \end{align*}
  Note that $h(1) = h'(1) = 0$ and $h''(1) \ge 0$. Thus, since it is clear that $h'''(\tau)\ge 0$ for all $\tau \in [1, \infty)$, we have that $h(\tau) \ge 0$ for all $\tau \in [1, \infty)$. Thus, (\ref{eq:tau}) follows, too. 
\end{proof}

Now we can prove the claim.

\begin{proof}[Proof of Claim~\ref{claim:submart}]
    Clearly if $Y_t = 0$, this is trivially true, so we now condition that $Y_t \neq 0$.
    
    Fix $x_t \in (0, 1)$ Let $f(x) = x^{1+\epsilon}$. By the convexity of $f$, we know that
    \begin{align*}
      f(x) &\ge f(x_t) + f'(x_t) (x - x_t)\\
      &= x_t^{1+\epsilon} + (1+\epsilon)x_t^{\epsilon}(x - x_t).
    \end{align*}
    Let $L_{x_t}(x) = x_t^{1+\epsilon} + (1+\epsilon)x_t^{\epsilon}(x - x_t)$. Thus,
    \begin{align}
      \mathbb E&[Y_{t+1} | Y_t, \hdots, Y_1 \wedge Y_t \in (0, 1)]\nonumber\\ &= \mathbb E[f(X_{t+1}) | X_t, \hdots, X_1 \wedge X_1 \in (0, 1)]\nonumber\\
                                                                  &= \mathbb E[f(X_{t+1}) - L_{X_t}(X_{t+1}) | X_t, \hdots, X_1 \wedge X_t \in (0, 1)] + \mathbb E[L_{X_t}(X_{t+1}) | X_t, \hdots, X_1 \wedge X_t \in (0, 1)]\nonumber\\
                                                                  &\ge \mathbb E[f(X_{t+1}) - L_{X_t}(X_{t+1}) | X_t, \hdots, X_1 \wedge X_t \in (0, 1)] + L_{X_t}(X_t) \quad \text{(by the submartingale property)} \nonumber\\
                                                                &= \mathbb E[f(X_{t+1}) - L_{X_t}(X_{t+1}) | X_t, \hdots, X_1 \wedge X_t \in (0, 1)] + Y_t.\label{eq:stuff}
    \end{align}
    Now as long as $X_t \in (0, 1)$, we have that
    \begin{align*}
      \frac{f(X_{t+1}) - L_{X_t}(X_{t+1})}{f(X_t)} &= \left(\frac{X_{t+1}}{X_t}\right)^{1+\epsilon} - (1+\epsilon)\frac{X_{t+1}}{X_t} + \epsilon.
    \end{align*}
    Note that the function $h(z) = z^{1+\epsilon} - (1+\epsilon)z + \epsilon$. Is strictly convex in $z$ and takes on minimum value at $z = 1$. We know that $\frac{X_{t+1}}{X_t} \not\in (1/\tau, \tau)$, we have that
    \[\frac{f(X_{t+1}) - L_{X_t}(X_{t+1})}{f(X_t)} = \min\left(\frac{1}{\tau^{1+\epsilon}} - \frac{1+\epsilon}{\tau} + \epsilon, \tau^{1+\epsilon} - (1+\epsilon)\tau + \epsilon\right) \ge \delta > 0,\]
    where we used Proposition~\ref{prop:calc}. In combination with (\ref{eq:stuff}), we have the desired inequality. 
  \end{proof}

\subsection{Proof of Claim~\ref{claim:beta}}

\begin{proof}
  Let $d_0 = 0$ and $c_{k+1} = 1$. Then, for all $i \in \{1, \hdots, k\}$ choose as our distribution $q$
  \[
    q_i = \frac{c_{i+1} - d_{i-1}}{2-\meas(E)}.
  \]
  In words, each $q_i$ is proportional to the length of the interval $[c_i, d_i]$ \emph{plus} the length of the surrounding ``space'' on either side. In particular, each interval of ``space'' is double-counted, so the sum of the $q_i$'s is $\frac{1}{2-\meas(E)}[\meas(E) + 2(1 - \meas(E))] = 1$. Thus, the $q_i$'s correspond to a well-defined probability distribution.

\noindent
  Next, observe that
  \begin{align*}
    \beta(E) &\ge \min(\mathbb E_{i \sim q}[U_0(i)], \mathbb E_{i \sim q}[U_1(i)])\\
             &= \frac{1}{2-\meas(E)}\min\left(\sum_{i=1}^k (c_{i+1} - d_{i-1})\prod_{j=1}^{i-1} \frac{1-c_{j+1}}{1 - d_j}, \sum_{i=1}^k (c_{i+1} - d_{i-1}) \prod_{j=1}^{k - i} \frac{d_{k-j}}{c_{k-j+1}}\right)
  \end{align*}
  Now define
  \[
    a_i = (c_{i+1} - d_{i-1})\prod_{j=0}^{i-1} \left( \frac{1-c_{j+1}}{1-d_j} \right) \quad \text{and} \quad 
    b_i = (c_{i+1} - d_{i-1})\prod_{j=0}^{k-i} \left( \frac{d_{k-j}}{c_{k-j+1}} \right)
    \]
  Note that we can include the $j = 0$ terms since $(1-c_1)/(1-d_0) = 1 = d_k/c_{k+1}.$ We now claim by induction that
\[
    \sum_{i=i_0}^k a_i = (1 - c_{i_0})\prod_{j=0}^{i_0 - 2} \frac{1-c_{j+1}}{1 - d_j} \quad \text{and} \quad 
    \sum_{i=1}^{i_0} b_i = d_{i_0}\prod_{j=0}^{k-i_0-1} \frac{d_{k-j}}{c_{k-j+1}}.
 \]
  For the first sum, the base case of $i_0 = k$ follows immediately, and for the second sum, the base case of $i_0 = 1$ also follows. Then, by the induction hypothesis
  \begin{align*}
    \sum_{i=i_0}^k a_i &= a_{i_0} + \sum_{i=i_0+1}^k a_i = (c_{i_0+1} - d_{i_0 - 1})\prod_{j=0}^{i_0-1} \frac{1-c_{j+1}}{1-d_j}+ (1 - c_{i_0 + 1})\prod_{j=0}^{i_0-1} \frac{1-c_{j+1}}{1 - d_j}\\
                       &= (1 - d_{i_0 - 1})\prod_{j=0}^{i_0-1} \frac{1-c_{j+1}}{1 - d_j} = (1 - c_{i_0}) \prod_{j=0}^{i_0 - 2} \frac{1 - c_{j+1}}{1 - d_j}.
  \end{align*}
  and similarly 
  \begin{align*}
    \sum_{i=1}^{i_0} b_i &= b_{i_0} + \sum_{i=1}^{i_0 - 1} b_i = (c_{i_0+1} - d_{i_0 - 1})\prod_{j=0}^{k-i_0} \frac{d_{k-j}}{c_{k-j+1}}  + d_{i_0 - 1}\prod_{j=0}^{k-i_0} \frac{d_{k-j}}{c_{k-j+1}}\\
                         &= c_{i_0+1}\prod_{j=0}^{k-i_0} \frac{d_{k-j}}{c_{k-j+1}} = d_{i_0}\prod_{j=0}^{k-i_0-1} \frac{d_{k-j}}{c_{k-j+1}}.
  \end{align*}
Thus, $\sum_{i=1}^k a_i = 1 - c_1 = 1$ and $\sum_{i=1}^k b_i = d_k = 1$. Therefore $\beta(E) \ge \frac{1}{2-\meas(E)}$, as desired.
\end{proof}

\section*{Acknowledgments}

We thank Brian Axelrod, Dima Kogan and anonymous reviewers for helpful feedback on the manuscript.

\bibliographystyle{alpha}
\bibliography{stoc19}

\end{document}